\documentclass[letterpaper, 10 pt, conference]{ieeeconf}  

\IEEEoverridecommandlockouts                              

\let\labelindent\relax
\usepackage{amsthm}
\usepackage{bbm}
\usepackage{lipsum}
\usepackage{comment}
\usepackage{amsfonts}
\usepackage{amsmath}

\usepackage{bm}
\usepackage[normalem]{ulem}
\usepackage{graphicx}
\usepackage{romannum}
\usepackage{amssymb}
\usepackage{enumitem}
\usepackage{xcolor}
\usepackage{cite}

\newcommand{\jb}[1]{\textcolor{teal}{[James: #1]}}

\newcommand{\terminalCost}[1]{\bar{l}^{#1}}

\newcommand{\beliefi}[2]{b^{#1}_{#2}}

\newcommand{\x}[1]{\text{\normalfont{\Romannum{#1}}}}

\newcommand{\belief}[1]{b_{#1}}
\newcommand{\game}[1]{\mathcal{G}^{#1}}
\newcommand{\e}{a}
\newcommand{\dprim}[1]{d(#1,g^1)}

\newcommand{\EAC}[1]{R^{#1}_{\sigma^*}}

\newcommand{\stagecost}[2]{l^{#1}_{#2}}
\newtheorem{theorem}{Theorem}
\newtheorem{lem}{Lemma}
\newtheorem{definition}{Definition}

\newtheorem{rem}{Remark}

\usepackage{import}
\newcommand{\Mover}{Attacker}
\newcommand{\Eater}{Defender}
\newcommand{\olstrategy}{{\gamma_\text{mix}}}
\newcommand{\olstrategyset}{{\Gamma_\text{mix}}}

\DeclareMathOperator*{\argmin}{arg\,min}

\renewcommand{\baselinestretch}{0.99}

\newenvironment{sketch}{\par\noindent\textit{Sketch of proof.}\ }{\hfill$\square$\par}

\usepackage{balance} 

\title{\LARGE \bf
Deceptive Path Planning: A Bayesian Game Approach
}

\author{Violetta Rostobaya$^1$, James Berneburg$^2$, Yue Guan$^2$, Michael Dorothy$^3$ and Daigo Shishika$^1$
\thanks{We gratefully acknowledge the support of ARL grant ARL DCIST CRA W911NF-17-2-0181. 
The views expressed in this paper are those of the authors and do not reflect the official policy or position of the United States Army, Department of Defense, or the United States Government.}
\thanks{This work has been submitted to the IEEE for possible publication.
Copyright may be transferred without notice, after which this version may no longer be accessible.}
\thanks{$^{1}$Violetta Rostobaya, James Berneburg and Daigo Shishika are with the Department of Mechanical Engineering,
        George Mason University,
        {\tt\small \{vrostoba,jbernebu,dshishik\}@gmu.edu}}%
\thanks{$^{2}$ Yue Guan is with the School of Aerospace Engineering at Georgia Institute of Technology,
        {\tt\small yguan44@gatech.edu}}%
\thanks{$^3$Michael Dorothy is with DEVCOM Army Research Laboratory, {\tt\footnotesize michael.r.dorothy.civ@army.mil}
}
}
\begin{document}

\makeatletter

\makeatother
\maketitle
\thispagestyle{empty}
\pagestyle{empty}

\begin{abstract}
This paper investigates how an autonomous agent can transmit information through its motion in an adversarial setting. 
We consider scenarios where an agent must reach its goal while deceiving an intelligent observer about its destination.
We model this interaction as a dynamic Bayesian game between a mobile Attacker with a privately known goal and a Defender who infers the Attacker’s intent to allocate defensive resources effectively. 
We use Perfect Bayesian Nash Equilibrium (PBNE) as our solution concept and propose a computationally efficient approach to find it. 
In the resulting equilibrium, the Defender employs a simple Markovian strategy, while the Attacker strategically balances deception and goal efficiency by stochastically mixing shortest and non-shortest paths to manipulate the Defender’s beliefs.
Numerical experiments demonstrate the advantages of our PBNE-based strategies over existing methods based on one-sided optimization.
\end{abstract}

\section{Introduction}
\noindent
\noindent
Deception---an act of concealing private information or conveying false information to gain an advantage---plays a pivotal role in strategic interactions across domains.
While much of the existing research has focused on deception through signaling in cyber domain~\cite{Huang2018DynamicBG, Kadem2024,Farhang2014ADB,YAVIN1987191,Basar2018}, this work investigates deception through motion~\cite{Dragan2015DeceptiveRM,shishika2024deception, rostobaya2023deception}.
Motion-based deception is distinct in that it introduces a fundamental tradeoff between goal-directed efficiency and the level of deceptiveness.
To resolve this tradeoff, this paper quantifies the effectiveness of motion-based deception against an intelligent observer through a game-theoretic lens.

We consider the scenario shown in Fig.~\ref{fig:scenario}, where a mobile agent, the Attacker, seeks to reach its true goal and carry out an attack, while its opponent, the Defender, attempts to infer the Attacker’s intentions and allocate defensive resources accordingly. 
The Attacker may employ deceptive movements such as approaching a fake goal to 
induce a misallocation of defenses, while the Defender plans its response with the possibility of deception in mind.
We model this interaction as a dynamic Bayesian game, where both players act strategically under asymmetric information.
\begin{figure}[t!]
    \centering
    \includegraphics[width = 0.5\textwidth]{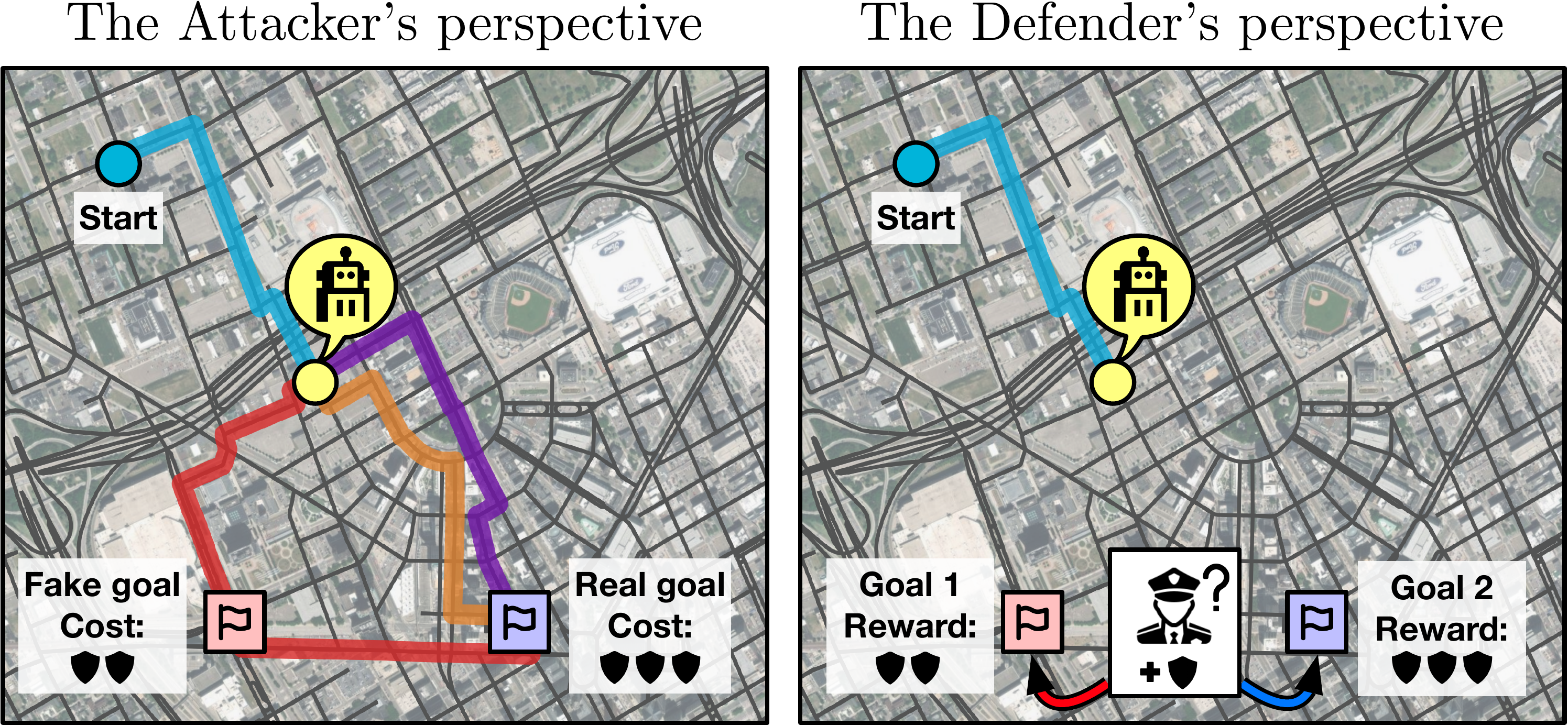}
    
    \caption{The \Mover{} at yellow node plans its path towards the true (blue) goal. Meanwhile, the \Eater{} infers the \Mover{}'s intentions and allocates defenses to one of two goals at each time step.  
    The accumulated defensive resources are indicated at each goal.}
    \vspace{-15 pt}\label{fig:scenario}
\end{figure}

The contributions of this work are: (i) the formulation of deceptive path planning as a dynamic Bayesian game; (ii)~a computationally efficient approach for obtaining a Perfect Bayesian Nash Equilibrium (PBNE) for this game; (iii) a comparative study demonstrating the effectiveness of the proposed game-theoretic strategies relative to existing deceptive path planning methods. Our results also demonstrate that both the appropriate level of deceptiveness and its efficacy are  dependent on the environment and the agent’s state.

\section{Related Work}
\noindent 
Goal-recognition scenarios have been widely used to study deceptive path planning and intent inference.
Prior works on intent inference assume that the mobile agent moves optimally towards its goal~\cite{Ramrez2010ProbabilisticPR, ziebart2008maximum,Pereira2017LandmarkBasedHF}.
These models are applied to deception~\cite{Dragan2015DeceptiveRM,Masters2017DeceptiveP,price2023domain}, typically framing it as a one-sided optimization where the deceiver knows the observer's inference model.
When this assumption is relaxed, critical questions arise: How should the agent plan its trajectory without access to the observer's inference model? Conversely, how should the observer interpret the agent's motion when the agent may act deceptively?
%
Moreover, the deceiver's objective is typically defined based on the observer's beliefs~\cite{Dragan2015DeceptiveRM,Masters2017DeceptiveP}, without accounting for how these beliefs influence the observer's actions.
Consequently, the appropriate level of deception for a given mission remains unclear.
Our work addresses these gaps using a game-theoretic framework.

Game theory has been previously used to study deception, 
but mainly in cybersecurity contexts~\cite{Huang2018DynamicBG, Kadem2024,Farhang2014ADB}, 
where false signals can be broadcast at low cost and without immediate physical consequences.
In contrast, our setting involves physically embodied agents whose actions are constrained by dynamics, energy, and environment. 
These physical constraints impose hard limits on the extent and nature of deception, requiring agents to balance deceptive behavior with goal-directed motion.
Our formulation explicitly captures this trade-off, distinguishing it from prior work in cyber domains.
The game formulated in~\cite{rostobaya2023deception} is closely related to ours, as it models a scenario similar to the one shown in Fig.~\ref{fig:scenario}. 
However,~\cite{rostobaya2023deception} formulates a non-zero-sum game in which the observing agent optimizes for worst-case payoff. 
In contrast, our formulation models the scenario as a zero-sum Bayesian game, where players optimize for expected payoff based on their beliefs. 
Additionally, while~\cite{rostobaya2023deception} limits the analysis to an obstacle-free grid world, our work generalizes the setting to arbitrary graph environments.


\if{false}
A prominent equilibrium concept for Bayesian games is the Perfect Bayesian Nash Equilibrium (PBNE), which strengthens the standard Bayesian Nash Equilibrium by incorporating consistent beliefs and sequential rationality. 
Existing methods for computing a PBNE include forward-backward iterative techniques~\cite{vasal2018systematic,aBHINAV2016}, sequential game decomposition~\cite{ouyang2016dynamic}, and normal form representations~\cite{Harsanyi1982}. 
However, in our problem the game ends when the deceiver arrives at its goal, which means the game can continue for arbitrarily long time, making set of possible game histories to be infinite. The complexity of our problem renders the mentioned conventional methods computationally intractable. Therefore, we propose a novel approach that leverages the unique structure of our formulation to compute a PBNE without resorting to iterative methods.
\fi{}

For Bayesian games, 
the PBNE 
is a standard solution concept.
Existing methods for finding PBNE include forward-backward iterative techniques~\cite{vasal2018systematic,aBHINAV2016}, sequential game decomposition~\cite{ouyang2016dynamic}, and normal form representations~\cite{Harsanyi1982}. 
However, these methods suffer from curse of dimensionality in games with large state spaces or extended horizons. 
To address this challenge, we propose a novel solution technique that leverages the structure of our problem to efficiently compute a PBNE without relying on iterative methods.


\section{Preliminaries}

\subsection{Problem Formulation}
\noindent
Consider a discrete-time dynamic game between the \Mover{} and the \Eater{} in a graph environment with two candidate goals. 
At the start of the game, Nature selects the true goal and privately reveals it to the Attacker, who is required to eventually reach it.
At each time step, the Defender receives a limited amount of defensive resources to allocate between the two goals, without knowing which one is the true goal.
These allocations are irreversible once made.
Consequently, the \Mover{} is incentivized to reach the goal efficiently, while also misleading the \Eater{} into selecting the fake goal, thereby minimizing the allocation of the resources to the true goal.

The environment is represented by a graph $(\mathcal{V},\mathcal{E})$, where $\mathcal{V}$ is a finite set of vertices and $\mathcal{E} \subset  \mathcal{V} \times \mathcal{V}$ is a set of edges. 
The edge weight $w(e)\in\mathbb{R}_{>0}$, for $e \in \mathcal{E}$  indicates the time needed for the Attacker to traverse an edge and the amount of resources the \Eater{} can allocate to a goal during this time of traversal.
 We let $g^1,g^2\in \mathcal{V}$ be the two candidate goals on the graph and assume that both are reachable from every node to avoid degeneracy.
We represent the private information---the true goal---as the Attacker's \emph{type} $\theta$,
modeled as a realization of a discrete random variable with finite support $\mathcal{I}:=\{1,2\}$ and a commonly known prior probability distribution $\beliefi{\theta}{0}\in (0,1)$.\footnote{Time indices are denoted in the subscript with $t$ or $k$, whereas the goal indices are denoted as superscripts with $i$ or $j$.} 
We denote the true goal of the \Mover{} type $\theta$ by $g^\theta$ and the fake goal by $g^{-\theta}$.


\paragraph*{States  and actions}
The game state $s_t\in \mathcal{V}$  represents the position of the \Mover{} at time $t$ and is controlled solely by the \Mover{}.
At each time step $t$, the \Mover{} acts first and transitions from $s_{t-1}$ to some node $s_t$ by selecting an action $e_t=(s_{t-1},s_t)$ 
from its action space $\mathcal{A}^A(s_{t-1})$ denoting the set of all outward edges from $s_{t-1}$.
After observing the Attacker's movement $e_t$, the \Eater{} selects one goal to allocate resources to.
The \Eater{}'s action $\e_t$ is thus drawn from $\mathcal{A}^D =  \{g^1,g^2\}$. 
Once resources are committed to the goal selected by the Defender, they cannot be reallocated to another goal later.

\paragraph*{Terminal condition} 
The game terminates when the \Mover{} reaches $g^\theta$ (and the \Eater{} takes its last action).
Note that if the \Mover{} reaches and then leaves the fake goal, the \Eater{} realizes the true goal, and the game reduces to a complete information game with a trivial solution (see Remark~\ref{rem. Complete information game.}).
Therefore, without loss of generality, we let the game terminate when the \Mover{} reaches either goal.
Formally, the game terminates at $T= \min_{t \in \mathbb{Z}_{\geq 0}} \{t \text{ }|\text{ } s_t \in \{g^1,g^2\}\}$, which is a random variable that depends solely on the strategy used by the \Mover{}. 

\begin{rem}[Complete Information Game]\label{rem. Complete information game.}
    When the \Eater{} knows \Mover{}'s type $\theta$, the optimal \Eater{} strategy is to allocate to $g^\theta$ at every time step, and the one for the \Mover{} is to move towards $g^\theta$ on a shortest path. The value of this complete-information game is 
    \begin{equation}\label{eq: comp.info.value}
        \bar U^\theta(s_0) = d(s_0,g^\theta),
    \end{equation}
    which is simply the distance to the true goal. 
\end{rem}


\paragraph*{Information structure} 
The common knowledge includes the graph, goal locations $\{g^1,g^2\}$, and the prior $\belief{0}=[b^1_0,b^2_0]$, whereas the \Mover{}'s type $\theta$ is a private information.
Both players have perfect recall of the state and action history $h_{t}=(s_0,e_1,s_1,\e_1,...,e_{t-1},s_{t-1},\e_{t-1})$, where the initial history is $h_1=(s_0)$.
We denote the set of admissible histories as $\mathcal{H}_t$ at each time step.
Since the \Mover{} acts first, the \Eater{} also has access to $e_t$ and the updated state $s_t$. 

\paragraph*{Strategy sets}

We let $\Sigma$ be the set of admissible behavioral strategies of the \Eater{}, and $\sigma\in\Sigma$ to be a specific \Eater{} strategy.
The \Eater{}'s policy at time $t$ is given by a mapping $\sigma_t: \mathcal{H}_t \times \mathcal{E} \times \mathcal{V} \rightarrow \Delta(\mathcal{A}^D(\cdot))$; here $\Delta(\mathcal{A})$ denotes an $|\mathcal{A}|$-dimensional simplex.
Policy $\sigma_t(g^i| h_{t},e_t,s_t)$ for $i\in\{1,2\}$ corresponds to the probability that the \Eater{} allocates to $g^i$.
Similarly, we use $\Gamma$ to denote the set of \Mover{}'s admissible behavioral strategies, and $\gamma\in\Gamma$ to be its element, which provides the \Mover{}'s policy at time $t$ given by a mapping $\gamma_t: \mathcal{H}_t \times \mathcal{I} \rightarrow \Delta(\mathcal{A}^A(\cdot))$.
The policy $\gamma_t(e_t|h_{t}, g^\theta)$ denotes the probability that the \Mover{} type $\theta$ takes edge $e_t$, given the history $h_t$.
To avoid degenerate cases, we further restrict $\Gamma$ such that $ \mathbb{P}^{\gamma}(T < \infty)=1$ for all $\gamma\in\Gamma$; i.e., all admissible \Mover{} strategies terminate the game in finite time.

\paragraph*{Beliefs}
We let $\beliefi{i}{t} \triangleq \mathbb{P}^{\gamma}(g^i|h_{t},e_t,s_t)$ be the \Eater{}'s belief that $g^i$ is the true goal at $t$, given its observed history ($h_{t},e_t,s_t$). The belief vector at time $t$ is $\belief{t}=[\beliefi{1}{t},\beliefi{2}{t}] \in \Delta (\mathcal{I})$. For Bayesian games, beliefs are assumed to be  propagated via Bayes' rule~\cite{Harsanyi1982}:
\begin{equation} \label{Bayes}
\beliefi{i}{t}=\frac{\beliefi{i}{t-1} \gamma_t(e_t|h_{t},g^i)}{\sum_{j \in \{ 1,2\}} \beliefi{j}{t-1}\gamma_t(e_t|h_{t},g^j)},\forall t>0.
\end {equation}
\paragraph*{Objective function}
We consider a zero-sum game in which the \Mover{} minimizes the expected final amount of resources allocated to its true goal, while the \Eater{} maximizes it. 
We let $\stagecost{\theta}{t}(e_t,\e_{t})$ be the stage cost:
\begin{equation*}
   \stagecost{\theta}{t}(e_t,\e_t)= 
   \begin{cases}
       w(e_t),& \text{ if } \e_t=g^\theta,\\
       0,&  \text{ otherwise}.
   \end{cases}
\end{equation*}
In words, a cost of $w(e_t)$ is incurred if and only if the \Eater{} allocates to the true goal $g^\theta$.
Additionally, we define a terminal cost function to handle cases where the game ends because the \Mover{} has visited the fake goal: 
\begin{equation*}
    \terminalCost{\theta}(s_{T}) \triangleq d(s_{T},g^\theta), 
\end{equation*}
where $d(s_{T},g^\theta)$ is the weighted distance from the final node $s_{T}\in\{g^1,g^2\}$ to the true goal $g^\theta$. 
Note that this will be zero if the game ends due to the \Mover{} reaching the true goal $g^\theta$, or will be the distance from the fake goal $g^{-\theta}$ to the true goal $g^\theta$, which is the cost accrued after $g^\theta$ is revealed to be the true goal. 
After the fake goal $g^{-\theta}$ is visited, the \Eater{} immediately knows $g^\theta$ is true and thus deterministically allocates to it, while the \Mover{} moves on a shortest path to $g^\theta$.

We use $U^\theta(\sigma,\gamma;s_0,\belief{0})$ to denote the type-specific cost representing the expected total allocation to $g^\theta$ at the end of the game:
\begin{equation}
\begin{split}
U^\theta(\sigma,\gamma;s_0,\belief{0})=\mathbb{E}^{\sigma,\gamma} \biggl [   \sum^{t=T}_{t=1} \stagecost{\theta}{t}(e_t,a_t) + \terminalCost{\theta}(s_{T})\biggl ] \\
    =\mathbb{E}^{\gamma} \biggl [ \sum^{t=T}_{t=1} w(e_t)\sigma_t(g^\theta|h_{t},e_t,s_t) + \terminalCost{\theta}(s_{T})\biggl ],
\end{split}
\end{equation}
where $\gamma$ determines $s_t$, $e_t$ and $T$.
The overall payoff is defined as follows: 
\begin{equation}
    U(\sigma,\gamma;s_0,\belief{0})=\sum_{i\in \mathcal{I}} \beliefi{i}{0} U^i(\sigma,\gamma;s_0,\belief{0} ), 
\end{equation}
which the \Eater{} aims to maximize and the \Mover{} aims to minimize.

\subsection{Perfect Bayesian Nash Equilibrium }
\noindent
The strategy profile $(\sigma^{*},\gamma^{*})$ constitutes a {Bayesian} Nash Equilibrium if it satisfies:
\begin{subequations}\label{eq:equilibrium} 
\begin{align}
    U(\sigma^{*},\gamma^{*};s_0,\belief{0})&\geq U(\sigma,\gamma^{*};s_0,\belief{0}),&\forall \sigma\in\Sigma, \label{Eq. BNE Eater} \\
U(\sigma^{*},\gamma^{*};s_0,\belief{0})&\leq U(\sigma^{*},\gamma;s_0,\belief{0}),&\forall \gamma\in\Gamma. \label{Eq. BNE Mover}
    \end{align}
\end{subequations}
A refinement that explicitly accounts for the dynamic aspect of the game is called PBNE 
defined below for our asymmetric information game.
\begin{definition}\label{def:PBNE}
Strategy profile $(\sigma^{*},\gamma^{*})$ constitutes a weak PBNE if, in addition to 
\eqref{Eq. BNE Mover}, it satisfies the following~\cite{Harsanyi1982}: 
    \begin{itemize}
        \item \textbf{Sequential rationality}: \Eater{}'s strategy $\sigma^*$ maximizes its expected payoff conditioned on its belief $b_t$.
        \item \textbf{Belief consistency}: Beliefs $\belief{t}$ are propagated using Bayes' rule in~\eqref{Bayes} with $\gamma^*$. Bayes' rule is applied wherever possible; otherwise, the beliefs remain unspecified.
    \end{itemize}
Note that the above two conditions subsume \eqref{Eq. BNE Eater}.
\end{definition}

In principle, our game can be represented in extensive form (i.e., as a game tree) and solved for PBNE using a backward-forward recursive method \cite{vasal2018systematic,aBHINAV2016,ouyang2016dynamic}.
 However, in our formulation, the \Mover{} can traverse the graph for an arbitrarily long time before reaching its true goal, leading to an infinitely large history set. 
 Even if we restrict Attacker's strategy set to strategies that reach the true goal in finite time, the resulting game tree remains prohibitively large---particularly in large graph environments---making direct computation infeasible. 
As a result, the direct use of existing PBNE solvers significantly limits the scalability (i.e., the size of the graph) of our deceptive path planning game.
\begin{figure}[t]
    \centering    \includegraphics[width = 0.45\textwidth]{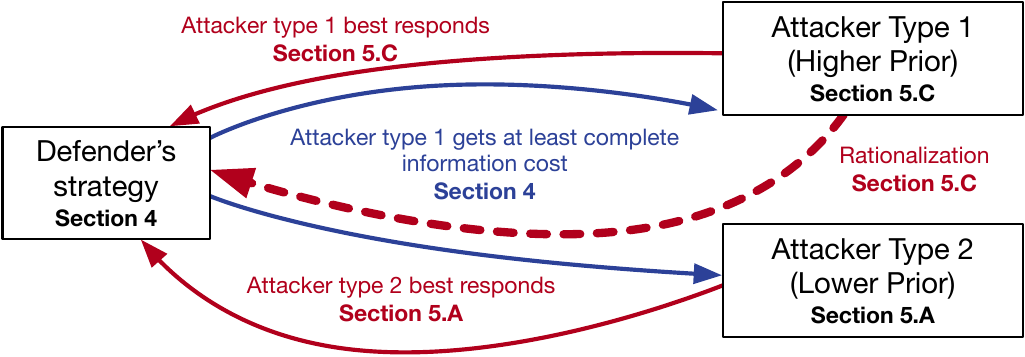}
    \caption{Overarching idea for how PBNE is established.}
    \vspace{-15pt}
\label{fig:idea}
\end{figure}

\subsection{Proposed Approach in Finding PBNE.}
\noindent
Figure~\ref{fig:idea} offers a high-level overview of our proposed approach. 
We begin with a simple \Eater{} strategy $\sigma^*$ that provides a lower bound on the expected cost for the \Mover{} type with the higher prior. 
We then construct the \Mover{}'s strategy $\gamma^*$ as the best response against $\sigma^*$.
However, this best response by itself does not provide an equilibrium since the \Eater{} may exploit it by deviating from $\sigma^*$.
Thus we additionally require the \Mover{}'s strategy to \textit{rationalize} the \Eater{}'s strategy, i.e., induce beliefs that make actions given by $\sigma^*$ to be sequentially rational. 
This constructive approach enables us to find a PBNE efficiently.


\section{The \Eater{}'s Strategy}

\noindent
This section presents \Eater{}'s strategy $\sigma^*$, whose optimality is proved later in Theorem~\ref{Th: PBNE} in Section~\ref{sec. Attacker Startegy}.
Before stating  $\sigma^*$, we define several auxiliary concepts.

\begin{definition}[Primary/Secondary Types]

We call the type with higher prior \emph{primary} and the other \emph{secondary}.\footnote{When the prior is $\beliefi{1}{0}=\beliefi{2}{0}=0.5$, the selection of the primary/secondary goal becomes arbitrary. Our results on PBNE still hold for each case.}
Without loss of generality, we let index $i=1$ (resp.~$i=2$) correspond to the primary (resp.~secondary): i.e., $\beliefi{1}{0}\geq\beliefi{2}{0}$.
\end{definition}

The Defender strategy we propose in the next section monitors the proximity of the Attacker to the primary goal, $g^1$, using the \emph{progress function} defined below.
\begin{definition}[Progress function] \label{def. c(t)}
Let $s_\tau$ be the \Mover{}'s position at  time $\tau$. We define $c_t$ as the minimum distance to the primary goal $g^1$ the \Mover{} has reached by $t$:
\begin{equation}
    c_t=\min_{\tau \leq t} d(s_\tau,g^1).
\end{equation}
\end{definition}

\noindent
Note that full state history is not required to update $c_t$. In practice, we can implement the following update rule:
\begin{align*}
    c_{t}&=\min \{c_{t-1},\dprim{s_{t}} \}, ~~\forall t>0,
\end{align*}
with the initial value $c_0=\dprim{s_0}$.


\noindent

We propose a Defender strategy that depends only on the \Mover{}'s most recent action $e_t$, updated state $s_t$ and progress function from the previous time step $c_{t-1}$.
By augmenting the state space with this progress function, such a strategy can be regarded as memoryless (Markovian).

\textbf{Defender's Strategy $\sigma^*$:}
The probability of allocating to $g^1$ is proportional to the progress towards $g^1$ that the \Mover{} has just made: i.e.,
\begin{equation} \label{eq: pi^E_A}
     \sigma_t^*(g^1| h_{t},e_t,s_t)=\sigma^*_t(g^1|c_{t-1},e_t)=\frac{c_{t-1}-c_t}{w(e_t)}\in[0,1]
\end{equation}
and $\sigma_t^*(g^2|c_{t-1},e_t)=1-\sigma_t^*(g^1|c_{t-1},e_t)$.\footnote{Note that $c_{t-1}$ is a sufficient statistic to implement $\sigma_t^*(g^i| h_{t},e_t,s_t)$ due to the Markovian property of $\sigma^*$. Also note that $\sigma^*$ is time invariant, and therefore we will drop the time subscript when appropriate.} 

The quantity $c_{t-1} - c_t\in[0,w(e_t)]$ measures the \emph{new} progress achieved towards $g^1$ at time step $t$.
Notice that even if $\dprim{s_t}<\dprim{s_{t-1}}$, the \emph{new} progress may be $c_t-c_{t-1}=0$, if the closest distance to $g^1$ was achieved earlier: i.e., $\dprim{s_\tau}\leq\dprim{s_t}$ for some $\tau<t-1$.
Also note that $c_{t-1} - c_t > 0$
if and only if $\dprim{s_\tau}<\dprim{s_t}$ for all $\tau<t$, i.e. $s_t$ is the closest node to $g^1$ among all previously visited nodes.
Furthermore, $c_{t-1} - c_t = w(e_t)$ if and only if the following two conditions hold: (i) $s_{t-1}$ was the closest point to $g^1$ among $\{s_\tau\}_{\tau=0}^{t-1}$, and (ii) $s_t$ is on a shortest path from $s_{t-1}$ to $g^1$.

%


Figure~\ref{Fig. Defender strategy}(a) shows an example of how $\sigma^*$ reacts to different actions by the \Mover{}.
As demonstrated in Fig.~\ref{Fig. Defender strategy}(b), this strategy ensures that the primary \Mover{} incurs a penalty precisely equal to the complete information game cost (see Remark~\ref{rem. Complete information game.}), which is formally stated in the next Lemma.
\begin{figure}[h!]
    \centering    \includegraphics[width = 0.49\textwidth]{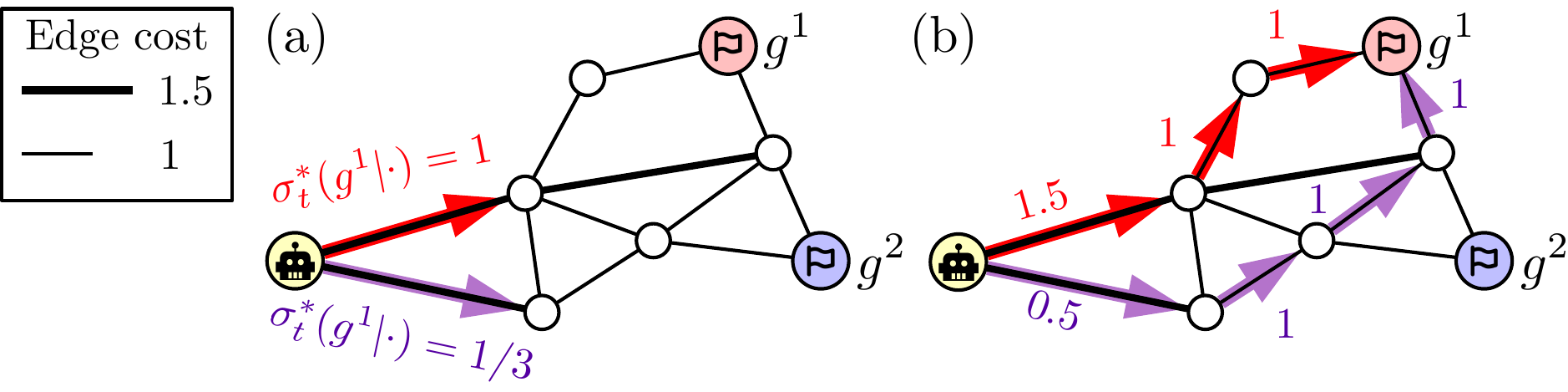}\label{Fig. Defender strategy}
\caption{Illustration of the Defender's strategy. (a)~The reaction of $\sigma^*$ to two different Attacker actions. The progress toward $g^1$ (red flag) is $c_0-c_1=3.5-2=1.5$ for the red action, whereas $c_0-c_1=3.5-3=0.5$ for the purple one. (b)~The expected penalty $w(e_t)\cdot\sigma^*(g^1|\cdot)$ along two different paths (red and purple) to $g^1$ demonstrating that both paths lead to the total cost of $U^1=d(s_0,g^1)=3.5$ in expectation.}
\end{figure}



\noindent 

\begin{lem}[Lower bound on $U^1$]\label{Lem: mover payoff in game A}
If the \Eater{} uses $\sigma^*$, then for any admissible $\gamma$, the primary \Mover{}'s cost is bounded by the initial distance to $g^1$: i.e.,
\begin{equation}\label{eq. Primary Game outcome}
    U^1(\sigma^*,\gamma;s_0,\belief{0})\geq \dprim{s_0}, \forall \gamma \in \Gamma,
\end{equation}
where the equality holds if $\mathbb{P}^{\gamma}(s_T=g^1|g^\theta=g^1)=1$.
\end{lem}

\begin{proof}
    We omit the proofs of lemmas that are not difficult to establish; full details are available upon request.
\end{proof}

The equality in Lemma~\ref{Lem: mover payoff in game A} comes from a ``hedging'' behavior of $\sigma^*$: i.e., the \Eater{} allocates to the less likely goal, $g^2$, whenever the \Mover{} is not making new progress towards $g^1$. 
 Consequently a wide variety of the \Mover{}'s strategies achieve the same cost $\bar U^1$ by avoiding $g^2$.
This ``slackness'' in the primary \Mover{}'s best response enables it to mix between actions and rationalize the \Eater{}'s actions, preventing its strategy from being exploited, as we will show in Section~\ref{sec. Attacker Startegy}.

\section{The \Mover{}'s Strategy}\label{sec. Attacker Startegy}

\begin{figure}[tb]
    \centering
    \includegraphics[width = 0.48\textwidth]{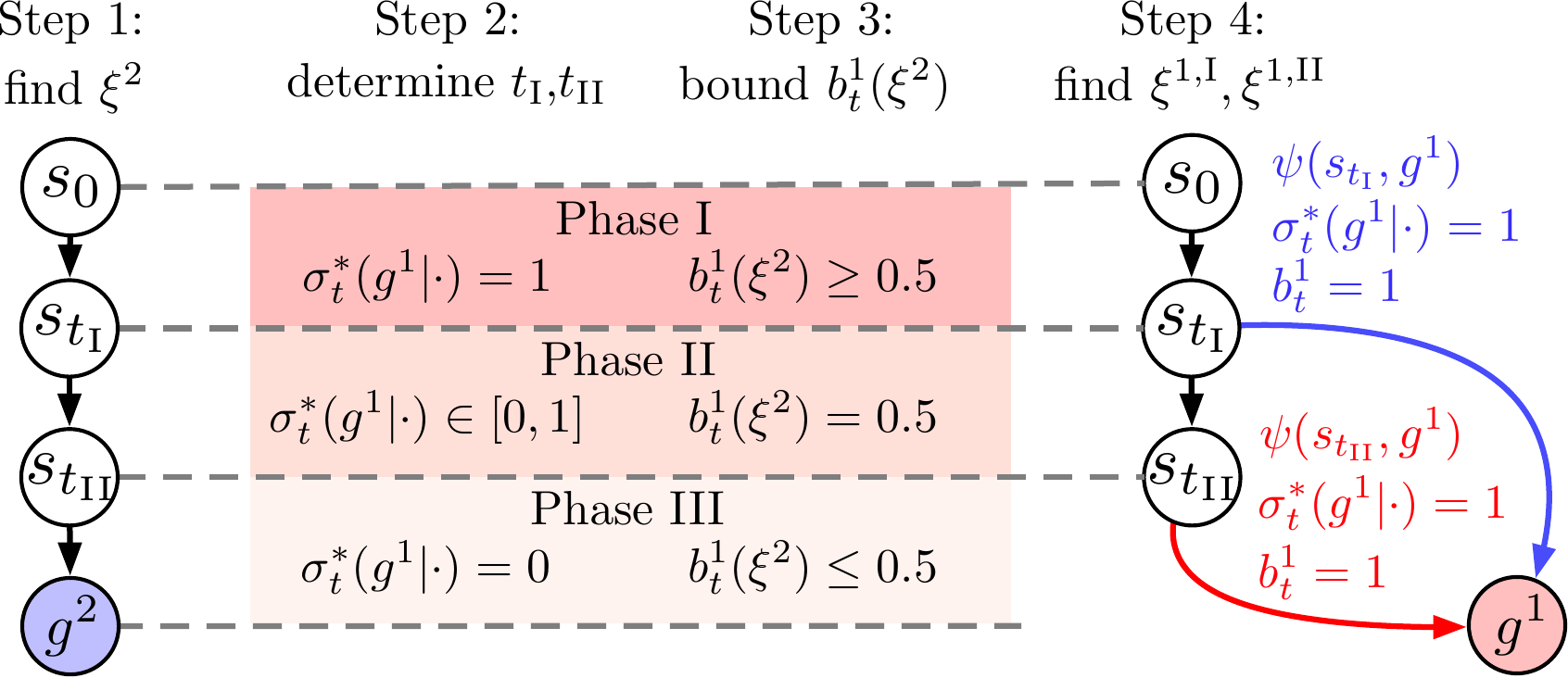}
\caption{Steps for finding \Mover{}'s strategy $\gamma^*$.}
    \label{fig:Mover's policy}
\end{figure}
\noindent
Figure~\ref{fig:Mover's policy} illustrates the construction of the Attacker strategy.
We first identify a path for the secondary \Mover{}, $\xi^2$, which is a best response path to $\sigma^*$ (Step~1). 
Based on how $\sigma^*$ acts along $\xi^2$, we segment $\xi^2$ into three phases (Step~2). 
We assign to each phase a belief that ``rationalizes'' the actions generated by $\sigma^*$ (Step~3). 
These beliefs are induced by constructing a primary \Mover{} strategy that makes probabilistic choices between continuing along $\xi^2$ or branching off from it (Step~4).

\subsection{Secondary \Mover{}'s strategy}\label{subsec:secondary}
\noindent
The secondary \Mover{}'s strategy is simply a best response to the \Eater{}'s strategy $\sigma^*$.
Although the space of the \Mover{}'s admissible strategies includes feedback/behavioral strategies, a best response to $\sigma^*$ can be found from a set of (open-loop) \emph{trajectories} to the goals. 
This is because the memoryless property of $\sigma^*$ allows the expected allocation to be fully determined for any segment,~$\varphi_{[k:K]}=(s_k,s_{k+1},...s_K)$, given the progress function at its start,~$c_{k}$.  
This expected allocation of resources under $\sigma^*$ over segment $\varphi_{[k:K]}$ is $\sum^{t=K}_{t=k+1}w(e_t^\varphi) \sigma_t^*(g^i|c_{t-1},e_{t}^\varphi)$, where $e_{t}^\varphi\triangleq\varphi_{[t-1:t]}$, and $c_t$ is updated along $\varphi$.
Furthermore, the following result allows us to reduce the search space from an infinite set of \emph{walks} (with loops) to a finite set of \emph{paths}, $\Xi$, that do not visit any node twice. 
\begin{lem}\label{lem: sec mover must use path}
If the \Eater{} uses $\sigma^*$, then the secondary \Mover{} has no incentive to visit any node more than once.
\end{lem}

For convenience, we introduce a subclass of open-loop strategies $\olstrategyset$ that maps Attacker's type to a distribution over paths originating from $s_0$: i.e., $\olstrategy:\mathcal{I}\rightarrow\Delta(\Xi)$. 
This remains consistent with our original problem formulation, as Kuhn’s equivalence theorem~\cite{Kuhn} ensures that for any $\olstrategy$, there exists an equivalent behavioral strategy~$\gamma\in\Gamma$.

\begin{definition}
Let $\Xi (s_0)$ be the set of paths from $s_0$ to either goal that avoid the other goal along the way,
and let $T^\xi$ be the terminal time for some path $\xi\in\Xi(s_0)$.
We use $\xi^2$ to denote a path that minimizes $U^2(\cdot)$ against $\sigma^*$:
\begin{equation} \label{eq: xi_B}
    \xi^{2} \in \argmin_{\xi \in \Xi(s_0)} \sum^{t=T^{\xi}}_{t=1}w(e_t^\xi) \sigma_t^*(g^2|c_{t-1},e_{t}^\xi) + \terminalCost{2}(s_{T^\xi}).
\end{equation}
 \end{definition}

\textbf{The secondary \Mover{}'s strategy:}
The secondary \Mover{} deterministically uses $\xi^2$: i.e., $\olstrategy^{*}(\xi^2|g^2)=1$.

\subsection{Game Phases}
\noindent
Let $\sigma^{*}_t(g^1|\xi^{2}_{[0:t]})$ denote the probability that the \Eater{} allocates to $g^1$ after observing the history corresponding to the \Mover{}'s trajectory $\xi^{2}_{[0:t]}$.
We segment $\xi^2$ into three ordered phases based on the evolution of $\sigma^*_t(g^1|\xi^{2}_{[0:t]})$ as defined below. 
While we impose a canonical ordering of Phases~$\x{1}$-$\x{3}$, we do not assume that all trajectories $\sigma^*_t$ naturally exhibit this temporal structure.
Instead, when the required structure does not arise, we treat the corresponding phase(s) as degenerate or of zero duration. This convention enables us to apply a unified analysis across all trajectories.


\paragraph*{\textbf{Phase I}} This is the initial phase during which the \Eater{} deterministically allocates to $g^1$, i.e., $\sigma^{*}_t(g^1|\xi^{2}_{[0:t]})=1, \forall t\in[1,t_{\x{1}}]$. 
The end of this phase can be found as
\begin{equation}\label{eq: t_1}
     t_{\x{1}}\triangleq
     \max \{t|\sigma^*_t(g^1|\xi^{2}_{[0:\tau]})=1,\forall \tau \leq t\}.
 \end{equation}
Phase I exists, i.e., has non-zero duration, if and only if $\sigma^{*}_t(g^1|\xi^{2}_{[0:1]})=1$; otherwise, $t_\x{1}=0$ if $\sigma^{*}_t(g^1|\xi^{2}_{[0:1]})<1$.

\paragraph*{\textbf{Phase II}}
This phase is the duration between the first time step the \Eater{} allocates to $g^2$ with non-zero probability, and the last time step at which the \Eater{} allocates to $g^1$ with non-zero probability.
The end of phase~II can be found as
\begin{equation}\label{eq: t_2}
    t_{\x{2}}\triangleq
        \begin{cases}
        T^{\xi^2},&\text{ if } \sigma^{*}_t(g^1|\xi^{2})>0,\\
        \min \{t| \sigma^{*}_t(g^1|\xi^{2}_{[0:\tau]})=0,\forall \tau > t \},&\text{ otherwise}.
        \end{cases}
\end{equation}
Phase II exists if and only if $t_{\x{1}}<t_{\x{2}}$.

\paragraph*{\textbf{Phase~\Romannum{3}}} During this final phase, the \Eater{} deterministically allocates to $g^2$, i.e., $\sigma^*_t(g^1|\xi^{2}_{[0:t]})=0,\forall t\in[t_\x{2}+1,T^{\xi^2}]$. 
Phase~\Romannum{3} exists if and only if $t_{\text{\Romannum{2}}} < T^{\xi^2}$.


The next section discusses the assignment of beliefs to the states along $\xi^2$.

\subsection{Sequential Rationalization}\label{subsec: primary}
\noindent
This section derives the primary \Mover{}'s strategy that satisfies the following conditions.

\begin{itemize}
\setlength{\itemindent}{-2em}
    \item[]\normalfont{(\textbf{C1})} The primary Attacker best responds to $\sigma^*$. 
   \item[]\normalfont{(\textbf{C2})} The combination of the primary \Mover{}'s strategy $\gamma^*_t(e_t|h_t,g^1)$ and the deterministic secondary \Mover{}'s strategy induces beliefs $b_t$ that rationalize $\sigma^*$.
\end{itemize}
Condition (\textbf{C1}) is met if the primary Attacker's cost is~\textit{no greater} than $\bar{U}^1(s_0)$, which is the case if it does not visit $g^2$ (see Lemma~\ref{Lem: mover payoff in game A}).
To see how condition (\textbf{C2}) can be satisfied, we start by presenting how our specific game formulation allows the \Eater{} to make rational/optimal decisions based solely on the belief at a given time step without considering future expected payoff. 

To this end, we consider a sub-class of \Mover{} strategies, $\widetilde{ \Gamma}$, that do not react to the \Eater{}'s action history. More formally, $\tilde \gamma: \widetilde{\mathcal{H}}_t \times \mathcal{I} \to \Delta(\mathcal{A}^A(\cdot))$, where $\tilde{h}_t\in\widetilde{\mathcal{H}}_t$ does not contain the \Eater{}'s actions: $\tilde{h}_{t}=(s_0,e_1,s_1,...,e_{t-1},s_{t-1})$.
Based on Kuhn's equivalence theorem, we can see that any mixed strategy $\olstrategy$, such as the secondary Attacker's strategy, can be equivalently represented by some $\tilde\gamma\in\tilde\Gamma$.
As we show later, the primary Attacker's strategy is also constructed in the set $\olstrategyset$, making the following result particularly relevant.

\begin{lem}\label{lem. greedy Eater}
For any $\tilde{\gamma} \in \widetilde{\Gamma}$, the \Eater{}'s strategy is optimal (sequentially rational) if and only if it is greedy with respect to its belief. 
Formally, given a belief, $b_t$, propagated according to~\eqref{Bayes} for an admissible history $(h_t, e_t, s_t)$, the \Eater{}'s optimal strategy must satisfy:

\begin{equation}
      \sigma^{*}_t(g^i|h_t, e_t, s_t) =1,  \quad \text{if  } ~ \beliefi{i}{t} > \beliefi{-i}{t}, 
\end{equation}
otherwise, if $\beliefi{1}{t} =\beliefi{2}{t} = 0.5$, any $\sigma^{*}_t(g^i|h_t, e_t, s_t) \in[0,1]$ is optimal.
\end{lem}

\begin{sketch}
Given an Attacker strategy $\tilde \gamma \in \tilde \Gamma$, it can be shown that $\sigma$ can be optimized without the dependency on its own action history: i.e., with $\tilde h_t$.  We can then consider the following dynamic program that optimizes $\sigma$:

 
    \begin{gather*}
        V_t^*(\tilde h_t, e_t, s_t) = 
        \max_{\sigma_t} \sum_{i\in \{1,2\}} b^i_t \Big[ \sigma_t(g^i|\tilde h_t, e_t, s_t) w(e_t) + \\ 
        \sum_{e_{t+1} \in \mathcal{E}(s_t)} \tilde{\gamma}^{i}_{t+1} V^*_{t+1}(\tilde h_{t+1}, e_{t+1}, s_t+e_{t+1})\Big],
  \end{gather*}
  where $\tilde{\gamma}^i_{t+1}\triangleq\tilde{\gamma}_{t+1}(e_{t+1}|\tilde{h}_t, e_t, s_t, g^i)$, and $s_t + e_{t+1}=s_{t+1}$ is the new state arrived after taking edge $e_{t+1}$.\footnote{Note that this value function is conceptually useful for proving the lemma, but solving the corresponding dynamic program is impractical because $b_t$ is a continuous variable.}
 Since the second term in square brackets is independent of $\sigma_t$, finding the optimal $\sigma^*_t$ reduces to  $\max_{{\sigma}_t} \sum_{i \in \{1, 2\}} \beliefi{i}{t} {\sigma}_t(g^i |h_t, e_t,s_t)$.
This reduction directly yields the result of the lemma.
\end{sketch}

To construct $\gamma^*_t(e_t|h_{t},g^1)$, the \Mover{} must consider beliefs on $\xi^2$ that will rationalize $\sigma^*$.
Let $\beliefi{1}{t}(\xi)$ denote the belief that $g^\theta=g^1$, after observing some trajectory $\xi_{[0:t]}$:
\begin{equation}\label{eq: belief on xi_B def}
    \beliefi{1}{t}(\xi)\triangleq \mathbb{P}^{\sigma^*,\gamma}(g^1|\xi_{[0:t]}).
\end{equation} 
The following lemma presents a restriction on beliefs that the \Mover{} can possibly induce on the states along $\xi^2$.
\begin{lem}[Monotonicity]\label{lem:monotonicity}
The belief $\beliefi{1}{t}(\xi^2)$ is a non-increasing function of time: i.e., 
\begin{equation} \label{eq: monotonicity of beta^xi_B}
   \beliefi{1}{t}(\xi^2)\leq \beliefi{1}{t-1}(\xi^2),\forall t\geq 1.
\end{equation}
\end{lem}

Based on Lemmas~\ref{lem. greedy Eater} and~\ref{lem:monotonicity}, we can bound the beliefs for each phase on $\xi^2$. For $\sigma^*$ to be rational along $\xi^2$, 
it is sufficient for $\beliefi{1}{t}(\xi^2)$ to satisfy the following:
\begin{itemize}
\item[]$\beliefi{1}{t}(\xi^2)\geq 0.5$, $\forall t\in[0, t_{\x{1}}]$ (Phase~$\x{1}$), 
\item []$\beliefi{1}{t}(\xi^2)=0.5$, $\forall t\in( t_{\x{1}}, t_{\x{2}}]$  (Phase~$\x{2}$),
and 
\item []$\beliefi{1}{t}(\xi^2)\leq 0.5$, $\forall t> t_{\x{2}}$  (Phase~$\x{3}$).
\end{itemize}
Based on these belief constraints, we construct primary \Mover{}'s strategy that satisfies conditions (\textbf{C1}) and (\textbf{C2}).


\begin{definition}[Primary \Mover{}'s paths]\label{def. attacker's paths}
Let $\psi(s,g^i)$ denote a shortest path from $s$ to $g^i$.
We use $\xi^{1,\x{1}}$ to denote a path constructed by connecting the phase~I portion of $\xi^2$, and a shortest path from $s_{t_\x{1}}$ to $g^1$: i.e.,
\begin{equation*}
    \xi^{1,\x{1}} \triangleq (\xi^2_{[0:t_\x{1}-1]}, \psi(s_{t_\x{1}},g^1)).
\end{equation*}$ $
The second path, $\xi^{1,\x{2}}$, is defined similarly as follows:
\begin{equation*}
    \xi^{1,\x{2}} \triangleq (\xi^2_{[0:t_\x{2}-1]}, \psi(s_{t_\x{2}},g^1)).
\end{equation*}
\end{definition}

Figure~\ref{fig. segments of paths} provides an example of a game in which all three phases exist on~$\xi^2$. Note that if phase~\Romannum{2} does not exist, i.e. $t_\x{1}=t_\x{2}$, then the two paths are identical $\xi^{1,\x{1}}=\xi^{1,\x{2}}$.
Also note that if $t_\x{1}=T^{\xi^2}$, then $\xi^{1,\x{1}}=\xi^2$.
Similarly, if $t_\x{2}=T^{\xi^2}$, then $\xi^{1,\x{2}}=\xi^2$.
Mixing between these two paths at a particular probability generates the desired beliefs that satisfy~(\textbf{C2}).

\begin{figure}[h!]
    \centering
    \includegraphics[width = 0.5\textwidth]{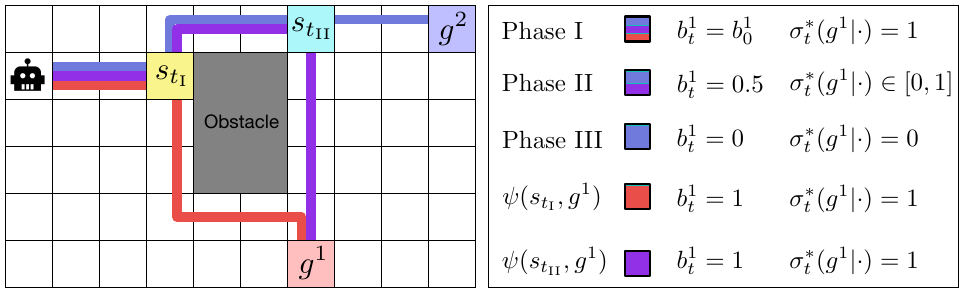}
    \caption{Paths used by the \Mover{}, when all three phases exist.}
    \label{fig. segments of paths}
\end{figure}

\paragraph*{\textbf{The primary \Mover{}'s strategy}}
The primary \Mover{} mixes between taking the paths $\xi^{1,\x{1}}$ and $\xi^{1,\x{2}}$ with probability ${\beliefi{2}{0}}/{\beliefi{1}{0}}$ and $1-{\beliefi{2}{0}}/{\beliefi{1}{0}}$:
\begin{equation}
\label{eq. Attacker's probabilities}
    \olstrategy^*(\xi^{1,\x{2}}|g^1)=1-\olstrategy^*(\xi^{1,\x{1}}|g^1)=\frac{\beliefi{2}{0}}{\beliefi{1}{0}}.
\end{equation}
Note that, for any $b^1_0\in[0.5,1)$, the equilibrium paths remain the same, although the mixing probability in~\eqref{eq. Attacker's probabilities} varies. 

Now we show what beliefs are induced by the overall $\gamma^*$.

\begin{lem}
\label{lem. induced beliefs}
     The \Mover{}'s overall strategy $\gamma^*$ induces following beliefs:
    \begin{itemize}
        \item []$\beliefi{1}{t}(\xi^2)=\beliefi{1}{0}$, $\forall t \in [0, t_\x{1}]$ (Phase~$\x{1}$), 
        \item[]$\beliefi{1}{t}(\xi^2)=0.5$, $\forall t \in ( t_\x{1}, t_\x{2}]$ (Phase~$\x{2}$), 
        \item[]$\beliefi{1}{t}(\xi^2)=0$, $\forall t \in ( t_\x{2}, T^{\xi^2}]$ (Phase~$\x{3}$),
    \end{itemize}
and $\beliefi{1}{t}(\xi^{1,\alpha})=1$ on the segment $\psi(s_{t_\alpha},g^1)$ after the branching from $\xi^2$, for $\alpha\in \{ \x{1},\x{2} \}$.
\end{lem}

Finally we state that overall Attacker's strategy $\gamma^*$ and the Defender's strategy $\sigma^*$ form a PBNE.
\begin{theorem}[PBNE]\label{Th: PBNE}
The strategy profile $(\sigma^*,\gamma^*)$, where $\sigma^*$ is defined in~\eqref{eq: pi^E_A}, and $\gamma^*$ is defined in~\eqref{eq: xi_B} and \eqref{eq. Attacker's probabilities}, constitutes a weak PBNE as defined in Definition~\ref{def:PBNE}.
\end{theorem}
\begin{sketch}
Based on Lemma~2 and \eqref{eq: xi_B}, the secondary \Mover{} has no incentive to deviate given $\sigma^*$ and any primary Attacker's strategy.
It can also be shown (omitted for page constraint) that the primary Attacker strategy described in Definition~\ref{def. attacker's paths} and \eqref{eq. Attacker's probabilities} achieves $U^1(\sigma^*,\gamma^*;s_0,b_0)=d(s_0,g^1)$  even when visiting the fake goal. 
This primary \Mover{}'s cost is optimal based on Lemma~1. Therefore, the Attacker has no incentive to deviate from $\gamma^*$.
By combining the results of Lemmas~\ref{lem. greedy Eater} and~\ref{lem. induced beliefs}, we see that the induced beliefs along all three equilibrium paths make the actions prescribed by $\sigma^*$ sequentially rational, thereby confirming that $\sigma^*$ is an equilibrium strategy.
\end{sketch}
Note that the \Eater{}'s beliefs for non-equilibrium \Mover{}'s strategies are not defined. However the \Eater{}'s strategy is well-defined for any possible history of \Mover{}'s actions and states.

\section{Discussion}\label{sec.:analysis of pbne}
\noindent

\noindent
{In general, a Bayesian Game may not have a PBNE. 
However, the constructive approach developed in this paper provides a solution for any given graph, goal locations, and initial position $s_0$, if the goals can be reached from any node. 
This indicates the existence of a PBNE for the deceptive path-planning game studied in this paper.}

A natural question to ask next is how our approach works for generic Bayesian games. We list some key structures that allow our method to work.

\emph{One-sided type uncertainty.}
Only the \Mover{} has a private type, and furthermore, the belief on this private type can be updated using information commonly available to both agents.
This structure ensures that both agents can keep a common belief~\cite{mahajan2012information}, and thus avoids the complexity of keeping a belief on the other agent's belief. 
        
\emph{Structure of dynamics and payoff.} 
The \Eater{} cannot directly control the state dynamics, and similarly, the \Mover{} cannot control the stage payoff. The payoff function also has the structure that represents resources allocation, i.e., the expected costs that the \Eater{} incurs on the two goals have a constant sum (regardless of the \Eater{}'s strategy). These properties allow Lemma~\ref{lem. greedy Eater}, which helps us simplify the assignment of beliefs for the rationalization. 

There are other simplifications that are made in this paper, including deterministic state transitions and the fact that there are only two candidate goals. The extension of the proposed procedure to relax these assumptions is the subject of ongoing work.

\noindent

\section{Numerical Experiments}
\noindent
To evaluate the effectiveness of the proposed strategies we use Value of Information defined below.
\begin{definition}[Value of Information (VoI).] We define VoI as the Defender's expected loss due to the lack of information with respect to the complete information value:
\begin{equation}
    \textnormal{VoI}(\sigma,\gamma;s_0,b_0) = \frac{\bar U(s_0,b_0) - U(\sigma,\gamma;s_0,b_0)}{\bar U(s_0,b_0)},
\end{equation}
    where $\bar U(s_0,b_0)\triangleq b_0^1 \cdot \bar U^1(s_0) + b_0^2 \cdot \bar U^2(s_0)$ is the complete information value over both games (see Remark~\ref{rem. Complete information game.}).
\end{definition}
For the same $s_0$ and $b_0$, the optimal strategies satisfying~\eqref{eq:equilibrium} will also optimize the VoI, where the \Eater{}/\Mover{} is the minimizer/maximizer.
We use $\text{VoI}^*\triangleq\text{VoI}(\sigma^*,\gamma^*)$ to denote the equilibrium $\text{VoI}$.
Importantly, for different environments and/or different $(s_0,b_0)$, VoI can be used as a metric to quantify the effectiveness of deception in a particular scenario.
\subsection{Illustrative example}\label{num sec. sec A}

\begin{figure}[t]
\centering
\includegraphics[width = 0.48\textwidth]{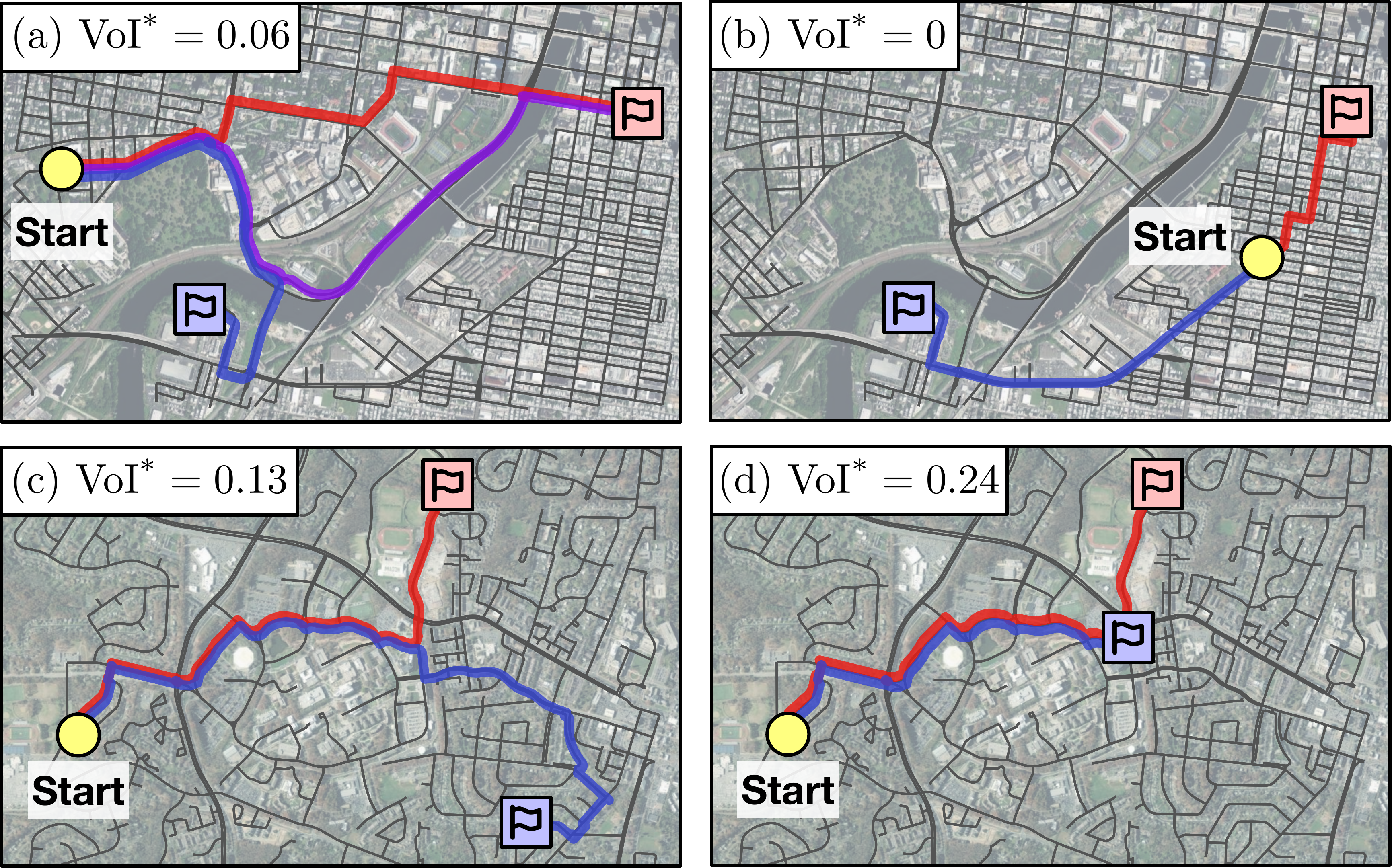} \caption{Equilibrium paths of the Attacker: (a)~all three phases exist along $\xi^2$ (blue path); (b)~only phase $\x{3}$ exists; (c)~only phases $\x{1}$ and $\x{3}$ exist; (d)~only phase $\x{1}$ exists.}
\vspace{-15pt}
    \label{fig. selling fig}
\end{figure}

\noindent
Figure~\ref{fig. selling fig} shows four examples of equilibrium paths and the corresponding $\text{VoI}^*$.
The prior is fixed at $b_0=[0.7,0.3]$ for all cases; however, $\text{VoI}^*$ varies based on the graph, location of goals, and the initial conditions.
The effect of the initial state is highlighted in Fig.~\ref{fig. selling fig}(a) and~\ref{fig. selling fig}(b). 
In~\ref{fig. selling fig}(a), the equilibrium \Mover{} maintains ambiguity by initially following $\xi^2$, resulting in a non-zero $\text{VoI}^*$.
In contrast, the optimal strategy in~\ref{fig. selling fig}(b) is a direct path to the true goal, suggesting no opportunity for deception from that state. 
Similarly, the influence of the goal locations are highlighted in~\ref{fig. selling fig}(c) and~\ref{fig. selling fig}(d).
In~\ref{fig. selling fig}(c) the presence of phase III is reflected in a lower $\text{VoI}^*$. 
In~\ref{fig. selling fig}(d), path $\xi^2$ fully overlaps with a portion of $\xi^1$, resulting in higher $\text{VoI}^*$.

\subsection{Baseline for Comparative Study}\label{num sec. sec B}
\noindent
We compare the equilibrium paths of the Attacker with two forms of deceptive motion: \textit{ambiguity}, where the agent obscures its true goal by simultaneously approaching both candidate goals; and \textit{exaggeration}, where the agent deliberately moves toward a fake goal to convince the \Eater{} that it is the intended target. The ambiguous paths ($\xi^{\theta,\text{am}}$) and exaggeration paths ($\xi^{\theta,\text{ex}}$) for each goal $g^\theta$ are generated using the one-sided optimization method from~\cite{Dragan2015DeceptiveRM}. 
In~\cite{Dragan2015DeceptiveRM}, the Defender assumes that the \Mover{} is acting optimally with respect to a presumed cost functional $C: \Xi \rightarrow \mathbb{R}_{\geq 0}$, defined over the space of trajectories $\Xi$. 
The \Eater{} then uses $C[\cdot]$ to form its belief $\mathbb{P}(g^i|\xi_{[0:t]})$ based on the observed path from $s_0$ to $s_t$:
 \begin{equation} \label{eq: belief generation anca dragan}
     \mathbb{P}(g^i|\xi_{[0:t]})=\frac{1}{Z}\frac{\text{exp}(-C[\xi_{[0:t]}]-V_{i}(s_t))}{\text{exp}(-V_{i}(s_0))}b^i_0,
 \end{equation}
 where $Z$ is a normalizer across the set of the candidate goals, $V_{i}(s)=\min_{\xi\in \Xi^{i}(s)}C[\xi]$, and $\Xi^{i}(s)$ is a set of paths from $s$ to $g^i$. 
 We use path length as the cost functional: $C[\xi_{[0:t]}]=\sum_{\tau=1}^{\tau=t}w(e^\xi_\tau)$, which gives $V_{i}(s)=d(s,g^i)$.

 The model in~\eqref{eq: belief generation anca dragan} is used to obtain exaggeration and ambiguous paths~\cite{Dragan2015DeceptiveRM}:
\begin{subequations}
 \begin{align}
     \xi^{\theta,\text{ex}} \in& \arg \min_{\xi \in \Xi^{\theta}(s_0)} \sum_{t=0}^{t=T^\xi} \mathbb{P}(g^\theta|\xi_{[0:t]}),  \\  
    \xi^{\theta,\text{am}}\in& \arg \min_{\xi \in \Xi^{\theta}(s_0)} \sum_{t=0}^{t=T^\xi} \lvert \mathbb{P}(g^\theta|\xi_{[0:t]})- \mathbb{P}(g^{-\theta} |\xi_{[0:t]} )\rvert.
\end{align}
\end{subequations}

To demonstrate the effectiveness of the equilibrium Defender $\sigma^*$, we compare it to an alternative strategy $\sigma'$ that uses~\eqref{eq: belief generation anca dragan} to generate a belief and acts greedily according to it (See Lemma~\ref{lem. greedy Eater}). When $b^1_t=0.5$, the alternative Defender deterministically allocates to $g^1$, since uniform belief makes any Defender's action rational. 

\subsection{Comparative study}\label{num sec. sec C}

\noindent
The Attacker is placed in $10 \times 10$ grid world environment with obstacles shown on Fig.~\ref{fig. anecdotal example paths}. The Attacker is allowed to move left, right, up and down, but not diagonally. 
The prior is $b_0=[0.6, 0.4]$. 

The primary Attacker mixes paths $\xi^{1,\x{1}}$ and $\xi^{1,\x{2}}$ with probabilities $\frac{1}{3}$ and $\frac{2}{3}$ respectively, while the secondary Attacker deterministically follows $\xi^{2}$. 
 \begin{figure}[h]\includegraphics[width = 0.48\textwidth]{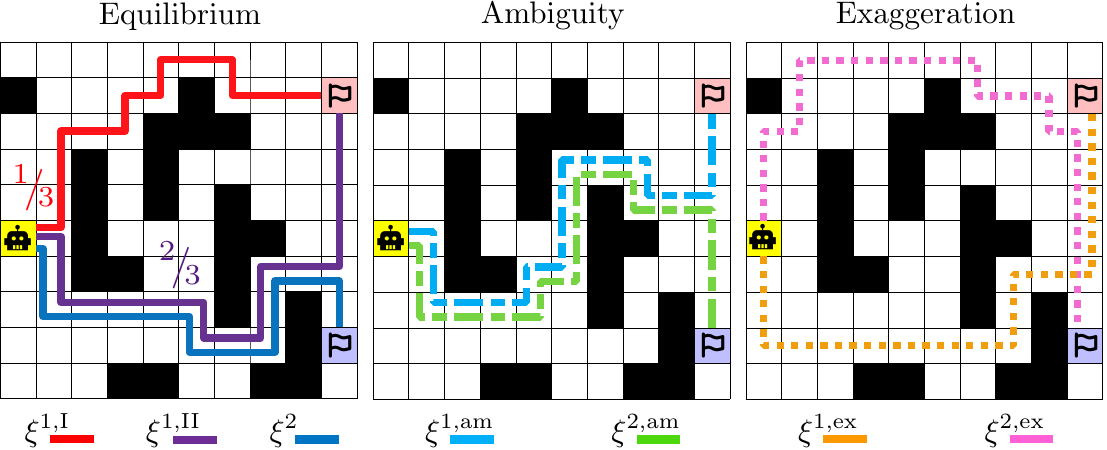}  
    \caption{Different deceptive behaviors of the Attacker: equilibrium, ambiguity and exaggeration paths. The red flag is $g^1$, the blue one is $g^2$, and the black squares are the obstacles.}
    \label{fig. anecdotal example paths}
\end{figure}

The corresponding accumulated expected allocation to $g^\theta$ under $\sigma^*$ over time is shown on Fig.~\ref{fig. anecdotal example plots}.
The left figure demonstrates that the primary Attacker cannot achieve any lower expected cost than $d(s_0,g^1)=15$, as was stated in Lemma~\ref{Lem: mover payoff in game A}. 
The right figure indicates the suboptimality of $\xi^{2,\text{am}}$ and $\xi^{2,\text{ex}}$.
\begin{figure}[h]
    \centering
    \includegraphics[width = 0.48\textwidth]{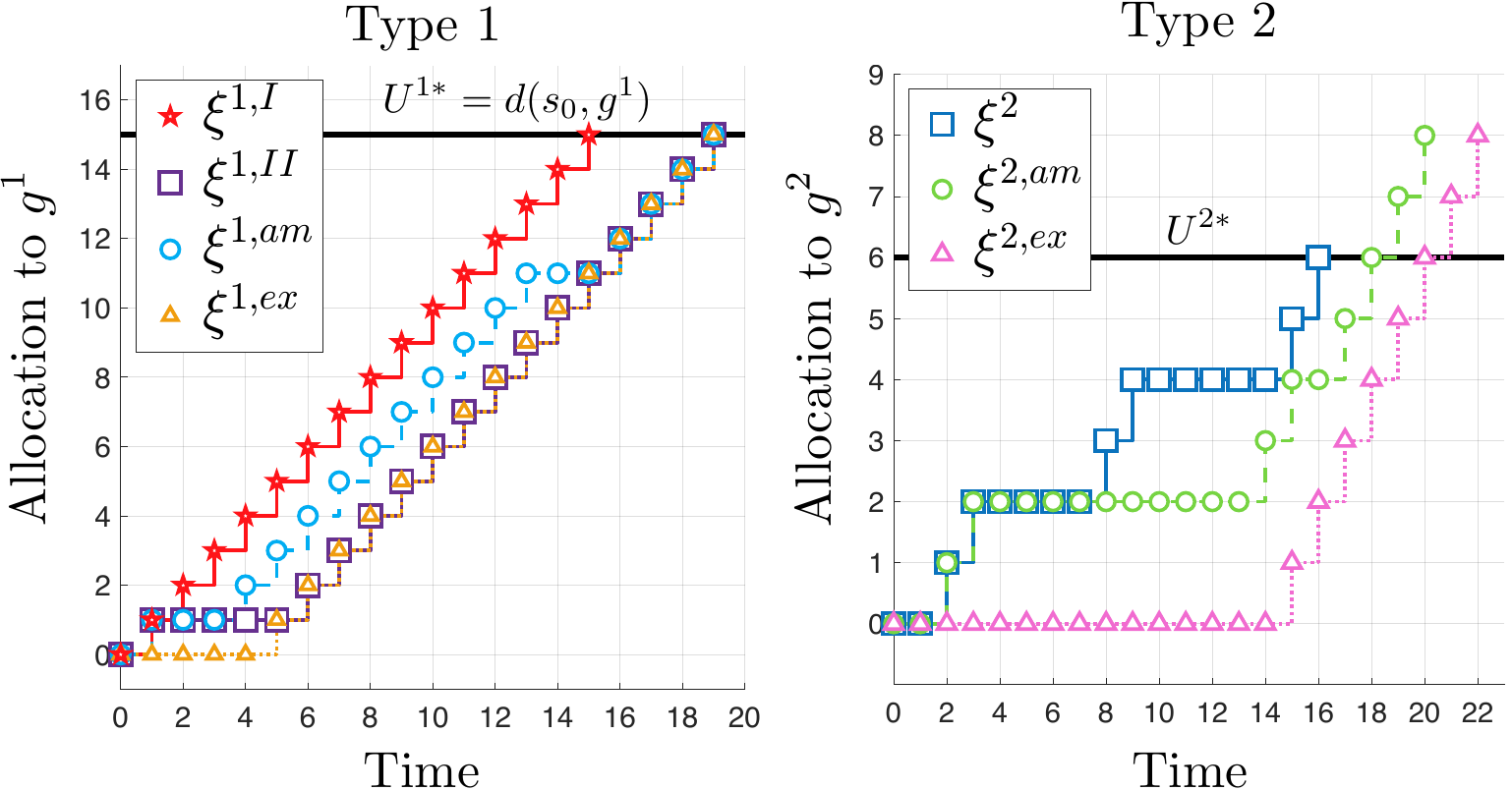}
    
    \caption{Corresponding accumulated expected allocation to $g^\theta$ over time.}
    \label{fig. anecdotal example plots}
\end{figure}
With the alternative Defender's strategy $\sigma'$, the expected cost for each Attacker type is $U^{1}(\sigma',\gamma^*)=\frac{1}{3}(15)+\frac{2}{3}(5)=8.33$ and $U^{2}(\sigma',\gamma^*)=15$. The overall expected payoff is $U(\sigma',\gamma^*)=11$, which is lower than equilibrium one $U^*(\sigma^*,\gamma^*)=11.4$.

\subsection{Statistical results}\label{num sec. sec D}
\noindent
We constructed $20\times20$ grid world graphs with obstacle densities varying from $0\%$ to $40\%$. 
 Here the obstacle density corresponds to the fraction of grid cells occupied by obstacles. 
All graphs share the same start and goal locations, as well as the prior of $b_0=[0.6,0.4]$. The obstacles are randomly generated in a way such that there is at least one path to each goal that does not lie through the other. 

\begin{figure}[t]
    \centering
    \includegraphics[width = 0.48\textwidth]{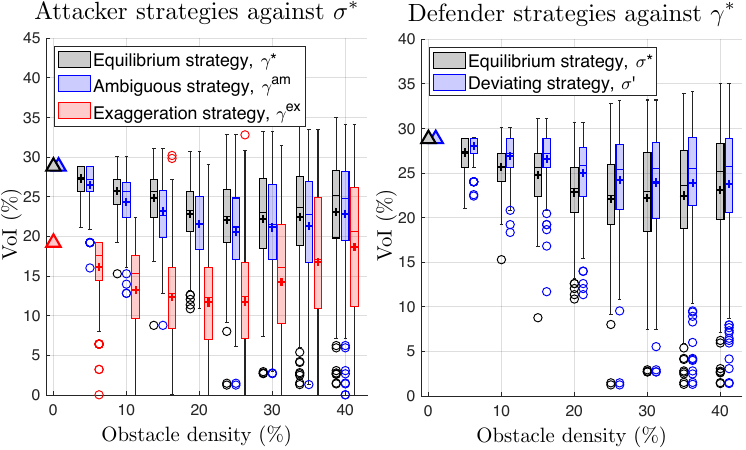}
    
    \caption{Deviation analysis in environments with different obstacle densities. The triangle sign indicates VoI in obstacle-free environment. 
    The plus/horizontal/circle sign indicates the mean/median/outlier.}
    \vspace{-15pt}
    \label{fig. VOI vs Density}
\end{figure}

Figure~\ref{fig. VOI vs Density} shows the statistical results of the equilibrium and deviating strategies. 
Each density has $200$ graph samples. 
The VoI against the equilibrium Defender is affected by the difference between $\bar U^2 (s_0)$ and $U^2 (\sigma^*,\gamma;s_0,b_0)$, because $\sigma^*$ ensures $U^1=\bar U^1(s_0)$ for any primary Attacker's strategy (see Lemma~\ref{Lem: mover payoff in game A}).
While both exaggeration and ambiguity strategies underperform relative to the equilibrium strategy, ambiguity results in a higher VoI, indicating it is the more suitable form of deception for the DPP game analyzed in this work.
The deviating Defender performs the same as the equilibrium one in an obstacle-free environment because the actions given by~$\sigma'$ are the same as those given by~$\sigma^*$.
Observe that the overall mean VoI decreases as obstacle density increases for both $\sigma^*$ and $\sigma'$, while the range of VoI increases.
This is because obstacles can induce either higher or lower VoI depending on how they are placed.
Higher VoI occurs when obstacles extend phases~\Romannum{1} and~\Romannum{2}. Conversely, VoI is lower when obstacles cause the paths to split earlier, resulting in shorter phases~\Romannum{1} and~\Romannum{2}, and a longer phase~\Romannum{3}.

\section{Conclusion}
\noindent
We present a Bayesian formulation of a DPP game to examine how an agent can regulate the information revealed through its motion. In our game, the \Mover{} must reach its designated target without revealing its true destination to the \Eater{}, who seeks to infer the goal and allocate defensive resources accordingly.
Recognizing the general difficulty of computing a PBNE, we propose a computationally efficient method that exploits the unique structure of our game to solve for the PBNE and find the equilibrium deceptive and counter-deceptive strategies that emerge in this setting.
Specifically, our approach starts with a construction of a memory-free \Eater{} strategy, which is then used to construct an \Mover{} strategy that not only optimally responds to this strategy but also induces beliefs that render the \Eater{}’s actions rational. As a result, the deceptive strategy for the \Mover{} involves a probabilistic mix between a goal-directed (shortest) path and a deceptive path that approaches a fake goal.
Future research directions include exploring additional game variations to assess the generalizability of our approach to a broader class of games.


\bibliographystyle{IEEEtran}
\bibliography{ref_aamas}

\end{document}